\documentclass{article}
\usepackage{ijcai23}

\usepackage{times}
\usepackage{soul}
\usepackage{url}
\usepackage[utf8]{inputenc}
\usepackage[small]{caption}
\usepackage{booktabs}
\usepackage{algorithm}
\usepackage{algorithmic}
\usepackage[switch]{lineno}
\usepackage{amsmath,amssymb}

\usepackage{amsthm}
\usepackage{xspace}
\usepackage{mathtools}
\usepackage[final]{changes}
\usepackage{caption}
\usepackage{tabularx}
\usepackage{graphicx}
\usepackage[colorlinks=true, allcolors=blue]{hyperref}
\usepackage{todonotes}
\usepackage{acronym}
\usepackage{xcolor}
\newif\ifuseboldmathops
\newif\ifuseittextabbrevs
\useboldmathopstrue   

\ifuseittextabbrevs

	\newcommand{\ie}{{\it i.e.}}

\else

	\newcommand{\ie}{i.e.~}

\fi

\ifuseboldmathops
	\newcommand{\reals}{\mathbb{R}}

\else
	\newcommand{\reals}{\mathbb{R}}

\fi

\ifuseboldmathops

\else

\fi

\ifuseboldmathops

\else

\fi

\ifuseboldmathops

\else

\fi

\newcommand{\argmax}{\mathop{\mathrm{argmax}}}
\newcommand{\argmin}{\mathop{\mathrm{argmin}}}

\newcommand{\truev}{\mathsf{true}}

\newcommand{\Always}{\Box \, }
\newcommand{\Eventually}{\Diamond \, }
\newcommand{\Next}{\bigcirc \, }
\newcommand{\until}{\mbox{$\, {\sf U}\,$}}

\newcommand{\sink}{\mathsf{sink} }

\newcommand{\calAP}{\mathcal{AP}}

\newcommand{\init}{{q_0}}

\newcommand{\last}{\mathsf{Last}}

\newcommand{\indicator}{\mathbf{1}}

\renewcommand{\vec}[1]{\mathbf{#1}}

\newcommand{\weight}{\mathbf{w}}

\theoremstyle{definition}
\newtheorem{definition}{Definition}
\newtheorem{example}{Example}

\newtheorem{lemma}{Lemma}

\newtheorem{theorem}{Theorem}

\newcommand{\calA}{\mathcal{A}}

\newcommand{\calP}{\mathcal{P}}

\newcommand{\dos}{\mathsf{d}}

\newcommand{\prioritizedAnd}{\&}

\acrodef{wdfa}[WDFA]{weighted deterministic finite automaton}
\acrodef{pqcltlf}[PQCLTL$_f$]{prioritized qualitative choice linear temporal logic on finite traces}
\acrodef{prefltl}[TPL]{temporal preference logic}
\acrodef{ltlf}[LTL$_f$]{linear temporal logic on finite traces}
\acrodef{qcl}[QCL]{qualitative choice logic}
\acrodef{pqcl}[PQCL]{prioritized qualitative choice logic}
\acrodef{pqcltl}[PQCTL]{prioritized qualitative choice temporal logic}
\acrodef{tlmdp}[TLMDP]{ labeled Markov decision process with a teminating state}
\acrodef{gpf}[GPF]{generalized preference formula}
\acrodef{cp}[CP]{ceteris paribus}
\acrodef{dfa}[DFA]{deterministic finite automaton}
\acrodef{mdp}[MDP]{Markov decision process}
\acrodef{ltl}[LTL]{linear temporal logic}
\acrodef{ldlf}[LDL$_f$]{linear dynamic logic on finite traces}
\acrodef{scltl}[syntactically co-safe LTL]{syntactically co-safe LTL}

\newcommand{\dist}{\mathcal{D}}

\renewcommand{\Pr}{\mathbf{Pr}}

\newcommand{\calM}{\mathcal{M}}

\newcommand{\lang}{\mathcal{L}}

\newcommand{\jf}[1]{\textcolor{blue}{JF: #1}}

\newcommand{\word}{\mathsf{Word}}

\newcommand{\trace}{\mathrm{trace}}

\usepackage{pgfplots}
\usepackage{subcaption}
\usepackage{paralist}
\usepackage{float}
\usepackage{comment}

\linenumbers

\newcommand{\orderedOR}{\overset{\rightarrow}{\times}}
\allowdisplaybreaks

\definechangesauthor{JF}
\definechangesauthor{LL}

\usepackage{tikz}
\usetikzlibrary{arrows,automata}

\newcommand*{\probleminternal}[4]{
	\par
	\medskip
	\noindent\fbox{\parbox{0.98\columnwidth}{
			\textbf{#4: #1} \\[0.05in]
			\renewcommand{\tabcolsep}{2pt}
			\begin{tabularx}{\linewidth}{rX}
				\emph{Input:} & #2 \\
				\emph{Output:} & #3
			\end{tabularx}
		}}
		\par
		\medskip
		\par
	}
	
\newcommand*{\problembox}[3]{\probleminternal{#1}{#2}{#3}{Problem}}

\newcommand{\PPwPPoTLO}{{\rm PPwPPoTLO}\xspace}

\title{Probabilistic Planning with Prioritized Preferences over \\ Temporal Logic Objectives}

\author{
Lening Li$^1$
\and
Hazhar Rahmani$^2$\and
Jie Fu$^2$
\affiliations
$^1$Lening Li is with the Department of Robotics, Worcester Polytechnic Institute, Worcester, MA 01609, USA.\\
$^2$Hazhar Rahmani and Jie Fu are with the Department of Electrical and Computer Engineering, University of Florida, Gainesville, FL 32605, USA.

\emails
lli4@wpi.edu,
\{h.rahmani, fujie\}@ufl.edu
 \thanks{This material is based upon work supported by Air Force Office of Scientific Research under award number FA9550-21-1-0085 and in part by NSF under award number 2024802.}
}

\begin{document}
\nolinenumbers

\maketitle
\thispagestyle{empty}
\pagestyle{empty}

\begin{abstract}
This paper studies temporal planning in probabilistic environments, modeled as labeled Markov decision processes (MDPs), with user preferences over multiple temporal goals. 
Existing works reflect such preferences as a prioritized list of goals. This paper introduces a new specification language, termed prioritized qualitative choice linear temporal logic on finite traces, which augments linear temporal logic on finite traces with prioritized conjunction and ordered \replaced{disjunction}{disconjunction} from prioritized qualitative choice logic. This language allows \added{for} succinctly specifying temporal objectives with corresponding preferences accomplishing each temporal task. The finite traces that describe the system's behaviors are ranked based on their  dissatisfaction scores with respect to the formula. We propose a systematic translation from the new language to a weighted deterministic finite automaton. Utilizing this computational model, we formulate and solve a problem of computing an optimal policy that  minimizes the expected score of dissatisfaction given user preferences. We demonstrate the efficacy and applicability of the logic and the algorithm on several case studies with detailed analyses for each.
\end{abstract}

\section{Introduction}
%
%

In this work, we study preference-based planning given  a preference order over temporal goals, \ie, ordered goals specified in  temporal \replaced{logics}{logic}. 
Temporal logic  are expressive and rigorous languages for specifying complex tasks and mission objectives.  
Planning with temporal logic goals \cite{pnueli1981temporal} \replaced{has}{have} seen  studied for robotic systems ~\cite{kantaros2022perception,bradley2021learning,he2020bp,vasile2020reactive,yang2020distributed,wang2020hyperproperties,hekmatnejad2018optimal,he2015towards,li2021reactive} and other intelligent systems~\cite{kasenberg2020generating,de2015synthesis,camacho2017non,mallett2021progression,zhou2022multiple,zhao2022temporal}. 

Specifying preferences  over temporal goals gives the decision maker flexibility to revise the task and \replaced{achieve}{achieves} the most preferred outcomes when not all constraints/subtasks can be satisfied.
 Early works consider deterministic systems---modeled as finite, discrete systems or system with deterministic dynamics. 
Several works \cite{tumova2013least,tuumova2013minimum,wongpiromsarn2021,vasile2017minimum} proposed minimum violation planning methods that decide which low-priority constraints should be violated in a deterministic system. 
\cite{mehdipourSpecifyingUserPreferences2021} associate weights with Boolean and temporal operators in signal temporal logic to specify the importance of satisfying the sub-formula and priority in the timing of satisfaction. They develop algorithms to maximize the weighted satisfaction in deterministic dynamical systems. 
\cite{rahmani2019optimal,rahmani2020you} studied temporal planning given both hard and soft specifications of the goal, using \ac{ltl} and \ac{ldlf}.
%
%
\cite{cai2020receding} consider minimizing the deviations from infeasible \ac{ltl}  specifications while maximizing  the total rewards.
Recently, preference-based probabilistic planning with temporal logic constraints are studied.
\cite{li2020probabilistic} study  preference-based planning for \ac{mdp} subject to an ordered list of probabilistic temporal logic formula. The algorithm enumerates  tasks one by one in a prioritized order until a policy that satisfies the most preferred  task is found.
\cite{Lahijanian2016} studied syntactically co-safe LTL planning with infeasible specifications in environments modeled by \ac{mdp}s. 
They compute a policy that maximizes the probability of satisfying a revised formula and minimizes the cost of revision.
\cite{lacerda2015optimal} considered a similar problem where the aim is to synthesize a policy that, in decreasing order of priority, maximizes the probability of completing the task, maximizes the probability of progressing toward completion, and minimizes the expected cost.


Despite the existing work on probabilistic preference-based planning, the connection between preference   specification in AI and preferences over temporal goals is yet to be established.  
We propose 
a new language that  extends temporal logic with  fuzzy logic representation of preferences. Specifically, we consider \ac{qcl} proposed in \cite{brewkaQualitativeChoiceLogic2004a}  and its extension \ac{pqcl} \cite{benferhatRevisedQualitativeChoice2007,benferhatTwoAlternativesHandling2008}. 
%
\ac{qcl} extends propositional logic with  a new logical connective called \emph{ordered disjunction}, denoted $\orderedOR$. Formula $A\orderedOR B$ means if possible then $A$, but if $A$ is not possible then at least $B$. 
  \ac{pqcl} 
  introduced prioritized conjunction/disjunction to \ac{qcl} by allowing the user to express priorities in a user's preference. Combined, if $(A\orderedOR B) \& (C\orderedOR D)$ where $\&$ is the prioritized conjunction, then the preference of $A\orderedOR B$ is more important to be satisfied than the preference $C\orderedOR D$. 

The  preference language proposed herein, called \ac{pqcltlf}, integrates \ac{pqcl} with a subclass of \ac{ltl} over finite traces. In particular, we introduce \ac{ltl} formulas for atomic preference and employ \ac{pqcl} to represent a preference over the temporal goals. We assign a dissatisfaction score for each outcome (temporal sequence of states) in light of semantics for \ac{pqcl}. This scoring function enables us to formulate a preference-based probabilistic planning objective, that is, to minimize the expected dissatisfaction score in a stochastic system, modeled as a labeled \ac{mdp}.

However, this compact, logical representation of preferences alone is not sufficient for probabilistic planning, which generally requires a computational model. Based on the relation between \ac{ltl} and automata, we developed a procedure that translates a \ac{pqcltl} formula into a \emph{weighted deterministic finite automaton}. This weighted automaton ensures for each path that satisfies the preference to a degree $k$, the sum of weights of the corresponding induced run on the weighted automaton is exactly $k$. Augmenting the planning state space with the state set of the weighted automata using a product operation, we show that the most preferred policy in the stochastic system can be obtained by solving a product \ac{mdp}. The correctness of the solution hinges upon the definition of a reward function based on the weights on transitions in the weighted automaton. We formally prove that the reward-maximizing policy in the product \ac{mdp} minimizes the expected degree of dissatisfaction in the original \ac{mdp} given the \ac{pqcltl} formula. In experiments, we employ several examples of robotic motion planning example to demonstrate the efficacy and applicability of the method and provide a detailed comparison of preference-based planning and traditional probabilistic planning with a monolithic temporal logic formula.

\section{Preliminaries}

\textbf{Notations: }
The set of all probability distributions over a finite set $X$ is denoted $\dist(X)$.



We introduce necessary preliminaries and notations next.

\paragraph*{System model} 
We model the interaction between the planning agent (a robot) and its stochastic environment as a variant of \ac{mdp}.

\begin{definition}[Labeled Markov Decision Process with a terminating state]
A \ac{tlmdp} is a tuple $M = \langle S, A \coloneqq \cup_{s \in S}A_s, P, s_0, s_{\bot}, \calAP, L \rangle$ where $S$ is a finite set of states; $A$ is a finite set of actions, where for each state $s \in S$, $A_s$ is the set of available actions at $s$;
{$A$ includes a special \emph{terminating action} $a_\bot$ and for any $s\in S$, $a_\bot \in A_s$.} 
$P \colon S \times A \times S \to [0, 1]$ is the transition probability function, where for each $s, s' \in S$ and $a \in A$, $P (s' \mid s, a)$ is the probability that the \ac{mdp} transitions to $s'$ after taking action $a$ at $s$; $s_0 \in S$ is the initial state; $s_{\bot} \in S$ is the \emph{terminating state}, which is a unique $\emph{sink}$ state. {For any $s\in S$, $P(s, a_\bot, s_\bot)=1$. That is, if an agent selects the terminating action $a_\bot$, then a terminating state $s_\bot$ can be reached surely.}
The set $\calAP$ is a finite set of atomic propositions; and $L \colon S \to 2^{\calAP} \cup \{\ltimes\}$ is a labeling function that assigns to each state $s \in S \setminus \{s_{\bot}\}$, the set of atomic propositions $L(s) \subseteq \calAP$ that hold in $s$. Only the terminating state is labeled the empty string, i.e., $L(s_\bot) = \ltimes$.

%
\end{definition}


A finite \emph{run} in this MDP is a sequence $\varrho = s_0 a_0 s_1 a_1 \cdots s_{k-1} a_{k-1} s_k$, in which, $s_0$ is the initial state and for each $0 \leq i \leq k-1$, $P(s_{i+1} \mid s_i, a_i) > 0$.
The path associated with this run is the sequence $\rho = s_0 s_1 \cdots s_k \in S^\ast$ an the \emph{trace} of this path is defined \added{as} $\trace(\rho) = L(s_0)L(s_1)L(s_2)\cdots L(s_k) \in (2^{\calAP})^\ast$.
A path \replaced{which}{who} ends at $s_{\bot}$ is called terminating.

A finite-memory, randomized policy in the \ac{mdp} is a function $\pi: S^{\ast} \rightarrow \dist(A) $ that maps a  state sequence into a distribution over  actions. A Markovian, or memoryless, randomized policy in the \ac{mdp} is a function $\pi: S  \rightarrow  \dist(A)$ that maps the current state into a distribution over actions. We denote the set of all Markovian randomized policies as $\Pi$.


A finite-memory, randomized policy $\pi:S^\ast  \rightarrow \dist(A)$ induces a Markov chain $M^\pi= \langle S^\ast,P^\pi \rangle$ over $S^\ast$ as follows: For any $\rho \in S^\ast$, $s\in S$, 
\begin{equation}
\label{eq:trans_induced}
  P^\pi(\rho s|\rho) =\sum_{a\in A} P(s|\last(\rho),a ) \cdot \pi(\rho, a), 
\end{equation}
where $\last(\rho)$ is the last state given the sequence $\rho$.


The stochastic process induced by a Markov policy is a Markov chain $M^\pi = \langle S^\ast,P^\pi\rangle$, where $P^\pi$  can be obtained as a special case of \eqref{eq:trans_induced}.
The probability of a path $\rho$ in the Markov chain $M^\pi$ is denoted by $\Pr(\rho; M^\pi)$.

\paragraph*{Planning objectives} 
We are interested in probabilistic planning subject to a preference over a set of temporal goals, captured using \ac{ltlf} formulas.  
%

\begin{definition}[LTL$_f$ Syntax \cite{de2013linear}]
Given a finite set $\calAP$ of atomic propositions,  the syntax of \ac{ltlf} formulas is defined as follows:
\[    
\varphi \coloneqq   p \mid \neg \varphi \mid \varphi  \land \varphi  \mid \Next \varphi \mid \varphi  \until \varphi,
\]
where $p \in \calAP$, negation ($\neg$) and conjunction ($\land$) are standard Boolean operators, and ``Next'' ($\Next$) and ``Until'' ($\until$) are temporal operators.

\end{definition}


Informally, formula $\Next \varphi$ states that $\varphi$ holds at the next time instant, and $\varphi_1 \until \varphi_2$ means there is a future time instant at which $\varphi_2$ holds and for all time instants from the current time until that future time, $\varphi_1$ holds true.
The temporal operator ``Eventually'' ($\Eventually$) is defined using ``Until'' as $\Eventually \varphi := \truev \until \varphi$. The dual of this operator is ``Always'' ($\Always$), which is defined as $\Always \varphi := \neg \Eventually \neg \varphi$. 
Formula $\Eventually \varphi$ means there is some future time instant at which $\varphi$ holds, while $\Always \varphi$ is interpreted that $\varphi$ is true at all future time instants. The semantics of \ac{ltlf} is given as interpretations over  finite traces and can be found in \cite{de2013linear}.
%

%


The language of an \ac{ltlf} formula $\varphi$, denoted $\lang(\varphi)$, is the set of finite traces over the alphabet $2^\calAP$ that satisfy $\varphi$. For notational simplicity, let $\Sigma \coloneqq 2^\calAP$ in the following context. The set of all finite words over a finite alphabet $\Sigma$ is denoted by $\Sigma^\ast$. The language of \ac{ltlf} formula $\varphi$ can be represented by the set of words accepted by a \ac{dfa} $\calA_{\varphi} = \langle Q, \Sigma, \delta, q_0, F \rangle$, where $Q$ is a finite set of states;  $\Sigma = 2^{\calAP}$ is a finite set of symbols called the alphabet; 
$\delta: Q \times \Sigma \rightarrow Q$ is a transition function such that $\delta(q,\sigma)=q'$ is the state reached upon reading input $\sigma$ from state $q$; $q_0\in Q$ is  an initial  state; and $F\subseteq Q$ is a set of accepting/final states. 
A transition function is recursively extended in the general way: 
$\delta(q,\sigma w)=\delta( \delta(q,\sigma),w )$ for given $\sigma\in \Sigma$ and $w\in \Sigma^\ast$. A word $w$ is \emph{accepting} if and only if $\delta(q, w)\in F$. The \ac{dfa} $\calA_\varphi$ accepts the exact set of words satisfying $\varphi$ given the semantics of \ac{ltlf}.


\section{Preference Language: Integration of Prioritized Qualitative Choice Logic and Temporal Logic}

In this section, we present a new task specification language to describe a subset of preferences over temporal goals. \replaced{We call this language}{This language, called} \emph{\ac{pqcltlf}}, which combines
\ac{ltlf} with \ac{pqcl}~\cite{benferhatRevisedQualitativeChoice2007}---a propositional logic for representing ranked objectives.

\ac{pqcl} augments propositional logic with a connective $\orderedOR$, called \emph{ordered disjunction}: A formula $\phi_1\orderedOR \phi_2$ means that if possible then $\phi_1$, and if $\phi_1$ is not possible then $\phi_2$. The operator $\orderedOR$ is left-associative,
and therefore $\phi_1 \orderedOR \phi_2 \orderedOR \ldots \orderedOR \phi_n =  \phi_1 \orderedOR (\phi_2 \orderedOR (\ldots \orderedOR \phi_n)\ldots)$.
%
%
In addition to ordered disjunction, \ac{pqcl} introduces \emph{prioritized conjunction}: A formula $\phi_1 \prioritizedAnd \phi_2$ defines the lexicographical ordering between individual satisfaction   of $\phi_1$ and $\phi_2$. 
%
%
%

\begin{definition}[Prioritized Qualitative Choice Linear Temporal Logic on Finite Traces]
  Let $\Phi$ be a set of \ac{ltlf} formulas over a set of atomic propositions $\calAP$.  A \ac{pqcltlf} fragment over $\calAP$ (without negation) is defined by  
  \[
    \varphi \coloneqq \psi \mid \varphi \orderedOR \varphi \mid \varphi \prioritizedAnd \varphi,
  \]
  in which $\psi \in \Phi$.
\end{definition}
In comparison to \ac{pqcl}, we do not include negation operation and thereby exclude the prioritized disjunction.
Negation is only allowed in the construction of \ac{ltlf} formulas.  The reason of not including negation is mainly due to ambiguity: A negation of $\text{flight}\orderedOR \text{train}$ can mean the two options are indifferent, incomparable, or train is preferred to flights. 


Different words satisfy a \ac{pqcltlf} formula to different degrees, which motivates us to quantify the number of alternatives a \ac{pqcltlf} formula can be satisfied.
%


\begin{definition}[Optionality, extended from~\cite{benferhatRevisedQualitativeChoice2007}]
 \label{def:optionality}
 Given an 
 \ac{pqcltlf} formula $\varphi$, the optionality of $\varphi$, denoted $opt(\psi)$, is the number of ways $\varphi$ can be satisfied, and is computed recursively as follows:
 \begin{itemize}
    \item If $\varphi$ is an \ac{ltlf} formula, then $opt(\varphi)=1$; 
    \item If $\varphi = \varphi_1 \orderedOR \varphi_2$, then $opt(\varphi)=opt(\varphi_1) + opt(\varphi_2)$;
    \item If $\varphi = \varphi_1 \prioritizedAnd \varphi_2$, then $opt(\varphi)=opt(\varphi_1) \cdot opt(\varphi_2)$.
  \end{itemize}
\end{definition} 
Associated with this definition of optionality, for each word $w\in \Sigma^\ast$ and a \ac{pqcltlf} formula, the word satisfies the formula to a certain degree.

\begin{definition}[Satisfaction Degree, extended from~\cite{benferhatRevisedQualitativeChoice2007}]
\label{def:satis_deg}
  Let $\varphi$ be a \ac{pqcltlf} formula over $\calAP$ and   $w\in \Sigma^\ast$ (recall $\Sigma = 2^\calAP$) be a finite word.  We write $w \models_k \varphi$ for some positive integer $k>0$ to 
  denote that the satisfaction degree of $w$ with respect to $\varphi$ is $k$, and use $w \not \models \varphi$ to denote that $w$ does not satisfy $\varphi$. 
 
   The satisfaction degree of $w$ with respect to $\varphi$ is computed as follows:
  \begin{itemize}
    \item If $\varphi$ is an \ac{ltlf} formula, then  $w \models_1 \varphi$ if $w \in \lang(\varphi)$, and
    $w \not \models \varphi$ if $w \not\in \lang(\varphi)$.
    \item If $\varphi = \varphi_1\orderedOR \varphi_2$, then $w\models_k \varphi_1\orderedOR \varphi_2$ if either
    \begin{itemize}
        \item $w\models_k \varphi_1$; or
        \item $w\models_n \varphi_2$, $w \not \models  \varphi_1$, and $k= n+opt(\varphi_1)$. 
        \end{itemize}
    \item If $\varphi = \varphi_1\prioritizedAnd \varphi_2$, then $w\models_k \varphi_1\prioritizedAnd \varphi_2$ if
    \begin{itemize}
    \item there exist $i, j > 0$ such that  $w\models_i \varphi_1$,   $w\models_j \varphi_2$, and $k= opt(\varphi_2)\times  (i-1)+j$; 
    \item otherwise,  if  $w \not \models \varphi_1$ or $w\not \models \varphi_2$, then $w \not \models \varphi_1\prioritizedAnd \varphi_2$.
    \end{itemize}
  \end{itemize}
  \label{def:degree_of_satisfaction}

  \end{definition}
    %
    %
  The definition of  satisfaction degree induces a total order only on the set of all the words that satisfy the \ac{pqcltlf} formula, but it does not rank those words who does not satisfy the formula.
    %
    %
 For planning purposes, we introduce a metric whose range of values is circumscribed between $0$ and $1$.
    
  \begin{definition}[Dissatisfaction Score]
  The dissatisfaction score function is a function  $\dos: \Sigma^\ast \times \Phi\rightarrow (0,1]$ that assigns to each word $w \in \Sigma^*$ and \ac{pqcltlf} formula $\varphi \in \Phi$, a positive real value in $(0,1]$, called the \emph{dissatisfaction score of $w$ with respect to $\varphi$}, which is computed as follows:
  \begin{itemize}
      \item If $w\not\models \varphi$, then $\dos(w,\varphi) = 1$;
      \item If $w\models_k\varphi$ for $k >0$, then $\dos(w,\varphi) = \frac{k}{opt(\varphi)+1}$.
  \end{itemize}
  \label{def:dissatisfaction_score}
  \end{definition}
  The lower the score, the more satisfied is the word. Note that the score is always greater than 0. 
  In the following context, when the formula $\varphi$ is clear from the context, we simply write $\dos(w)$ for the   dissatisfaction score of $w$ w.r.t. $\varphi$.

  %

%


\paragraph*{Preference model generated from a \ac{pqcltlf} formula}
Each \ac{pqcltlf} formula $\varphi$ over a set of atomic propositions $\calAP$ induces a preference model $\succeq^\varphi$ over $\Sigma^*=(2^\calA)^*$ such that
for any two words $w,w'\in \Sigma^\ast$, $w$ is preferred to $w'$ with respect to $\varphi$, \ie,  $w\succeq^\varphi w'$, if and only if 
$\dos(w,\varphi)\le \dos(w', \varphi)$.

It is easy to prove the following property.
\begin{lemma} If $w \succeq^\varphi w'$, then
one of the following conditions holds:
\begin{inparaenum}
\item  $w\models_k \varphi $ and $\replaced{w'}{w} \not \models \varphi $; or 
\item $w\models_n \varphi $, $\replaced{w'}{w} \models_m \varphi$, and
$n \le m$.
\end{inparaenum} 
\end{lemma}

The preference model over $\Sigma^\ast$ directly translates to a preference model over $S^\ast$---the set of finite paths in a labeled \ac{mdp} such that path $\rho \in S^\ast $ is preferred to $\rho' \in S^\ast$ if and only if $\dos(L(\rho), \varphi)\le\dos(L(\rho'), \varphi) $.
Thus, given a policy in a labeled \ac{mdp}, we introduce the following measure to evaluate how preferred a policy is with respect to a \ac{pqcltlf} formula. 

\begin{definition}[Expected  Dissatisfaction Score]
Let $\pi$ be a finite-memory, randomized policy for a given MDP,  $M^\pi= \langle S^\ast,P^\pi \rangle$ be its induced Markov chain, and $\varphi$ be a \ac{pqcltlf} formula. The \emph{expected dissatisfaction score} of $\pi$ with respect to $\varphi$, denoted by $\dos(\pi, \varphi)$, is defined
  \begin{equation}
  \label{eq:policy_sat_degree}
     \dos(\pi,\varphi)
     = \sum_{\rho \in S^\ast} \Pr(\rho;M^\pi)\cdot \dos(L(\rho),\varphi).
  \end{equation}
\end{definition}

 

We now formally state the probabilistic planning problem:

\problembox{Probabilistic Planning with Prioritized Preferences over Temporal Logic Objectives (\PPwPPoTLO)}
{A   labeled \ac{mdp} $M = \langle S, A \coloneqq \cup_{s \in S}A_s, P, s_0, s_{\bot}, \calAP, L \rangle$ and a \ac{pqcltlf} formula $\varphi$.}
{
A policy $\pi \colon S^\ast \rightarrow\dist(A)$ that minimizes the expected dissatisfaction score of $\varphi$.
\label{prob}
}

\section{Optimal planning for PQCLTL$_f$ formulas}
We now present a planning algorithm to solve the \PPwPPoTLO problem. Our approach consists of two steps: In the first step, we construct an automata-theoretic model for \ac{pqcltlf} formula. In the second step, we show that the optimal policy that minimizes the expected dissatisfaction score of the given formula can be computed by solving a reward-maximizing \ac{mdp} with  augmented states.

\subsection{Automata-theoretic modeling of PQCLTL$_f$ formulas}
We focus on constructing a computational model for a given \ac{pqcltl} formula $\varphi$. The choice of such a computational model for representing the subclass of \ac{pqcltlf} formulas is a \emph{weighted deterministic finite automaton}. 
\begin{definition}[Weighted Deterministic Finite  Automaton~\cite{droste2009weighted}]
  A \emph{weighted deterministic finite automaton} is a tuple $\calA = \langle Q,\Sigma, \delta, q_0, \weight \rangle$,
  where $Q$ is a finite set of states; $\Sigma \cup \{\ltimes\}$ is a finite set of symbols (alphabet); and $\ltimes$ is a unique symbol representing the end of a string \footnote{In general, one can include $\rtimes$ as the beginning of a finite string and $\ltimes$ as the ending of a finite string. The beginning symbol $\rtimes$  is omitted as it is clear from the context.}; $\delta: Q\times (\Sigma \cup \{\ltimes\}) \rightarrow Q$ is a deterministic transition function; $q_0$ is the initial state; and  $\weight: Q\times (\Sigma \cup \{\ltimes\}) \times Q\rightarrow \reals$ is a weight function that assigns each transition $(q,\sigma, q')$ to a real value, called the weight of this transition.
\end{definition}
Consider a finite word   $w=\sigma_0 \sigma_1 \ldots \sigma_{n-1}\ltimes$,    let $w[i]$ be the $i$-th symbol of this word.  The run $\rho$ generated by word $w$ is $\rho \coloneqq q_0 \sigma_0 q_1 \ldots \sigma_{n-1} q_n $ that satisfies $q_{i+1}=\delta(q_i, w[i])$, for $i=0,\ldots, n-1$.
We write $\word(\rho)=w$ to denote the word associated with the run $\rho$. 
%
The total weight is $\weight(\rho)= \sum_{i=0}^{n-1} \weight(q_i,w[i],q_{i+1})$.

First, we show how to construct the \ac{wdfa} for an \ac{ltlf} $ \varphi$.

\begin{definition}[\ac{wdfa} for an \ac{ltlf} Formula $\varphi$]
\label{def:wdfa_ltlf}
Let $\calA_\varphi = \langle Q,\Sigma, \delta, q_0, F \rangle $ be a \ac{dfa} encoding $\varphi$. A \ac{wdfa} for encoding $\varphi$ is constructed from $\calA_\varphi$ as a tuple 
\[
  \calA = \langle Q \cup \{\sink\},\Sigma \cup \{
  \ltimes\}, \delta', q_0, \weight \rangle
\]
in which for each $q \in Q \cup \{\sink\}$ and $\sigma \in \Sigma \cup \{ \ltimes \}$,
\begin{equation}
    \delta'(q, \sigma) = 
    \begin{cases}
    \delta(q, \sigma) & \text{if } q \neq \sink \text{ and } \sigma \neq \ltimes \\
    \sink  & \text{otherwise,}
\end{cases}
\end{equation}
and for each $q, q' \in Q \cup \{\sink \}$ and $\sigma \in \Sigma \cup \{ \ltimes \}$, 
\begin{equation}
    \weight(q, \sigma, q') = 
    \begin{cases}
    1 & \text{if } q \in F  \text{ and } \sigma = \ltimes \text{ and } q' = \sink \\
    0  & \text{otherwise.}
\end{cases}
\end{equation}
\end{definition}

Intuitively, the \ac{wdfa} $\calA$ extends the DFA $\calA_\varphi$ with a sink state $\sink$. For any  state $s\in S$ of the original \ac{dfa} $\calA_\varphi$, a transition to $\sink$ is made with an input symbol $\ltimes$. A weight one is received only if the transition is from an accepting state to the sink state upon reading the ending symbol $\ltimes$. 
%
\begin{lemma}
  Given a \ac{wdfa} $\calA$ for an \ac{ltlf} formula $\varphi$ and a finite run $\rho =q_0 \sigma_0 q_1 \ldots \sigma_{n-1}q_n$, if $\weight(\rho)=1$ then $\word(\rho)\models_{1} \varphi$.
  \label{lma:ltlf}
\end{lemma}
The proofs of   Lemmas \ref{lma:ltlf}, \ref{lma:orderedDisjunct} \ref{lma:prioritizedConj}, and \ref{lma:option} can be found in Appendix~\ref{app:proofs}.


Next, we define the construction process of \ac{wdfa}s for ordered disjunction and prioritized conjunction of \ac{pqcltlf} formulas. 

\begin{definition}[\ac{wdfa} for Ordered Disjunction of \ac{pqcltlf} Formulas]
\label{def:wdfa_orderOR}
Let $\calA_i = \langle Q_i\cup \{\sink_i\}, \Sigma\cup \{\ltimes\}  , \delta_i, \init_i, \weight_i \rangle$ for $i=1,2$ be two  \ac{wdfa}'s that respectively encode two \ac{pqcltlf} formulas $\varphi_1,\varphi_2$.
One can constrcut from them, a \ac{wdfa} for $\varphi_1\orderedOR \varphi_2$ as a tuple $\calA  = \langle Q_1\times Q_2 \cup \{ \sink \}, \Sigma \cup \{\ltimes\} , \delta,( \init_1, \init_2),  \weight \rangle$, in which, the transition function is defined as,  for any $(q_1,q_2)\in Q_1\times Q_2$ and $\sigma\in \Sigma \cup \{ \ltimes \}$, 
    \begin{align*}
        \delta((q_1, q_2), \sigma) = 
        \begin{cases}
            (\delta_1(q_1, \sigma), \delta_2(q_2, \replaced{\sigma}{q})) & \text{if }\sigma \neq \ltimes \\
             \sink  & \text{otherwise,}
        \end{cases}  
    \end{align*}
    and the weight function is defined as, 
    \begin{itemize}
  \item For any $(q_1, q_2) \in Q_1 \times Q_2$, input $\sigma\in \Sigma$
            \[
            \weight((q_1, q_2), \sigma, (\delta_1(q_1, \sigma), \delta_2(q_2, \sigma))) = 0
            \]
            \item For any $(q_1, q_2) \in Q_1 \times Q_2$, input $\ltimes$, 
            \begin{multline*}
        \weight((q_1, q_2), \ltimes, \sink) = \\
        \begin{cases}
            \weight_1(q_1,\ltimes, \sink) & \mbox{if }\weight_1(q_1,\ltimes,\sink)>0 \\
            \weight_2(q_2,\ltimes, \sink)+ opt(\varphi_1) &  \text{if $\weight_1(q_1,\ltimes,\sink)=0$ }\\
            & \text{and $\weight_2(q_2,\ltimes, \sink)>0$} \\
             0  & \text{otherwise,}
        \end{cases}    
    \end{multline*}
 \end{itemize}
\end{definition}

\begin{lemma}
  \label{lma:orderedDisjunct}
  Given a \ac{wdfa} $\calA$ for  $\varphi =  \varphi_1\orderedOR \varphi_2$ and a finite run $\rho =\vec{q}_0 \sigma_0 \vec{q}_1 \ldots  \vec{q}_{n-1} \sigma_{n-1} \vec{q}_n$, if $\weight(\rho)=k$ for a $k > 0$, then $\word(\rho)\models_{k} \varphi$, else $\word(\rho) \not \models \varphi$.
\end{lemma}

\begin{definition}[\ac{wdfa} for Prioritized Conjunction of \ac{pqcltlf} Formulas]
\label{def:wdfa_prioritizedAnd}
Let $\calA_i = \langle Q_i\cup \{\sink_i\}, \Sigma\cup \{\ltimes\}  , \delta_i, \init_i, \weight_i \rangle$ for $i=1,2$ be two  \ac{wdfa}'s that respectively encode two \ac{pqcltlf} formulas $\varphi_1,\varphi_2$.
One can construct from them, a 
 \ac{wdfa}
for $\varphi_1 \prioritizedAnd \varphi_2$ as a tuple
$\calA  = \langle Q_1\times Q_2 \cup \{ \sink \}, \Sigma \cup \{\ltimes\} , \delta,( \init_1, \init_2),  \weight \rangle$, in which, the transition function is defined as, for any $(q_1,q_2)\in Q_1\times Q_2$ and $\sigma \in \Sigma \cup \{ \ltimes \}$,
\begin{align*}
        \delta((q_1, q_2), \sigma) = 
        \begin{cases}
            (\delta_1(q_1, \sigma), \delta_2(q_2, \replaced{\sigma}{q})) & \text{if } \sigma \neq \ltimes ,\\
             \sink  & \text{otherwise.}
        \end{cases}    
    \end{align*}
and the weight function is defined as, 
\begin{itemize}
  \item For any $(q_1, q_2) \in Q_1 \times Q_2$, for $\sigma \in \Sigma$, 
            \[
            \weight((q_1, q_2), \sigma, (\delta_1(q_1, \sigma), \delta_2(q_2, \sigma))) = 0\]
            \item   For any $(q_1, q_2) \in Q_1 \times Q_2$, for input $\ltimes$, if $\weight_i(q_i,\ltimes,\sink) >0$ for both $i=1,2$, then 
          \begin{multline*}
                \weight((q_1, q_2), \ltimes, \sink)  = \weight_2(q_2,\ltimes,\sink)        \\ + opt(\varphi_2) \cdot (\weight_1(q_1,\ltimes,\sink)-1),
          \end{multline*}   else
 $         \weight((q_1, q_2), \ltimes, \sink) =0$.
 \end{itemize}
\end{definition}


\begin{lemma}
  \label{lma:prioritizedConj}
  Given a \ac{wdfa} $\calA$ for  $\varphi =  \varphi_1\prioritizedAnd \varphi_2$, and a finite run $\rho =\vec{q}_0 \sigma_0 \vec{q}_1 \ldots  \vec{q}_{n-1} \sigma_{n-1} \vec{q}_n$, if $\weight(\rho)=k$ for a $k > 0$, then $\word(\rho)\models_{k} \varphi$, else $\word(\rho) \not \models \varphi$.
\end{lemma}

Given the above construction methods of \ac{wdfa}s for \ac{pqcltlf} formulas, the \ac{wdfa} for a more complex \ac{pqcltlf} formulas can be constructed recursively.

\begin{lemma}
\label{lma:option}
  Given a \ac{pqcltlf} formula $\varphi$ for which $\lang(\varphi)\ne \emptyset$ and the constructed \ac{wdfa} $\calA$, the \replaced{optionality}{optionally} of $\varphi$ is the maximal weight of all transitions in $\calA$. That is 
  \[
  opt(\varphi) = \max\{\weight(q,a,q')\mid \delta(q,a,q') \text{ is defined.}\}
  \]
\end{lemma}
An example to illustrate the construction of \ac{wdfa}s is given in the Appendix ~\ref{app:example}.
 \section{Probabilistic Planning to Minimizing Dissatisfaction Score}
 
In this section, we show how to leverage the \ac{wdfa} for solving Problem~\PPwPPoTLO. Similar to probabilistic planning with linear temporal logic constraints, a product operation between the labeled \ac{mdp} and the \ac{wdfa} allows us to keep track of temporal objectives.

\begin{definition}[The product between the labeled \ac{mdp} and a   \ac{wdfa}]
\label{def:prodMDP}
  The product of a given
   \ac{wdfa} $\calA = \langle Q \cup \{\sink\},\Sigma \cup   \{\ltimes\}, \delta, \init, \weight \rangle$ and a terminating labeled  \ac{mdp} $M = \langle S, A \coloneqq \cup_{s \in S}A_s, P, s_0, s_{\bot}, \calAP, L \rangle$ is an MDP
  \[
    \calM= M\otimes \calA =( V 
    , A := \bigcup_{v \in V} A_v , \calP, v_0 , R)
  \]
  in which
  \begin{itemize}
    \item $V = S\times Q  $ is the state space,
    \item $A$ is the set of actions, and for each $v=(s, q) \in V$, $A_v = A_s$ is the $v$' set of 
    available actions,
    \item  $\calP$ is the probabilistic transition function, where for each states $(s, q), (s', q') \in V$ and action $a \in A$,
          \begin{align*}
            \calP((s,q), a, (s',q'))= 
            P(\replaced{s}{s_1},a,\replaced{s'}{s_2}) \cdot \indicator(\delta(q,  L(s'))= q')
          \end{align*}
    \item $v_0=(s_0,   \delta(\init, L(s_0)))$ is the initial state.
    \item $R: V\times A\rightarrow \reals$ is the reward function, where for each $(s,q)\in V$ and $a\in A$, if $a=a_\bot$ and  $\weight(q,\ltimes, \sink)>0$, then $R((s,q), a_\bot) =  opt(\varphi) -  \weight(q, \ltimes, \sink)+1$, else $R((s,q), a ) = 0$. 
  \end{itemize}
\end{definition}

Given  a finite run $h = v_0a_0v_1a_1\ldots v_n$ in the product \ac{mdp}, the total reward is $R(h)= \sum_{i=0}^{n-1} R(v_i, a_i)$. Since a run $h$ in the product \ac{mdp} corresponds to a run $\rho$ in the original \ac{mdp} except that each state in $\rho$ is augmented with an automaton state, we use $\mathsf{Proj}_S (h)$ to compute the projection of the run $h  = (s_0, q_0) (s_1,q_1)\ldots (s_n, q_n) \in V^\ast$ to a run $s_0 s_1s_2\ldots s_n \in S^\ast$ whose labeling is $L(s_0s_1\ldots s_n) =  L(s_0)L(s_1)\ldots L(s_n)$. 
We denote the set of finite runs in $\calM$ by $Runs(\calM)$.
%

%
Based on the reward function,   the expected total reward of a  {nonstationary} policy $\pi: V^\ast \rightarrow \dist(A) $ for an initial state $v\in V$ is defined as
\[
  J_\pi(v)= \lim\sup_{N\rightarrow \infty} J_{\pi,N}(v),
\]
with $J_{\pi, N}(v)$ being the expected $N$-stage
reward of $\pi$ for state $v$:
\[
  J_{\pi,N}(v) = E \left [ \sum_{t=0}^{N} R(V_k, \pi(V_0\ldots V_k))\mid V_0= v \right],
\]
where $V_k$ is the state at time $k$.  The expectation  is with respect to the probability distribution of  paths in Markov chain $\calM^\pi$. 

\begin{lemma}
  For any  policy $\pi: V^\ast \rightarrow \dist(A)$ of the product \ac{mdp} $\calM$, for any $v\in V$, $J_\pi(v) < \infty$.
\end{lemma}
The proof is in Appendix~\ref{app:proofs}.

The optimal value function is defined to be 
\[ 
    J^\ast(v)  =\arg\max_{\pi} J_\pi(v).
\]
%
%
For optimal planning to maximize the total reward, $J^\ast(v)$ can be attained by a Markovian policy \cite{puterman2014markov}. 
%
Therefore, in
 the following, we only consider Markovian policies. 
 We also consider the Bellman operator $T$, defined by
\begin{align*}
    T J(v) = &\max_{\pi \in \Pi}  \sum_{a\in A_v}[ R(v,a)\pi(a \mid v) \\ &+ \sum_{v' \in V}\calP(v'|v,a)\pi(a \mid v)J(v\added{'})],
\end{align*}
and the optimal value function satisfies $TJ^\ast = J^\ast$.

Among all the Markovian policies for the product MDP, we consider only the \emph{proper} ones.
%


 \begin{definition} [Extended from \cite{bertsekas2013stochastic}]
  A policy $\pi$ for the \ac{mdp} $\calM$ is \emph{proper} if it guarantees that the sink state
  $(s_{\bot}, \sink)$ will be reached with probability one.
\end{definition}

\begin{lemma}
  The optimal value $J^\ast(v)$ for any $v\in V$ can be obtained by a proper, Markovian policy of   product \ac{mdp} $\calM$.
\end{lemma} The proof is in Appendix~\ref{app:proofs}.

Thus, to search an optimal policy, we need to consider only proper, Markovian policies. We now relate the reward maximizing problem in the product \ac{mdp} to the planning objective of minimizing the expected dissatisfaction score.



\begin{lemma}
\label{lem:d_and_R}
  For each path $h  = (s_0, q_0) (s_1,q_1)\ldots (s_n, q_n) \in V^\ast$, it holds that,
  \begin{multline*}
            \dos(L(\mathsf{Proj}_S (h)), \varphi) = 
      \begin{cases}
     1-\frac{R(h)}{opt(\varphi)+1} & \text{if } R(h) > 0,  \\
      1  & \text{if } R(h) = 0 
     .\end{cases}
  \end{multline*}
\end{lemma}The proof is in Appendix~\ref{app:proofs}.

\begin{theorem}
  Let $\pi: V \rightarrow \dist(A)$ be a policy for the product MDP $\calM$.
  Construct from $\pi$, a policy $\pi': S^\ast \rightarrow\dist(A)$ for $M$ such that for each
  $\rho: s_0 s_1 \cdots s_n \in S^*$, $\pi'(\rho) = \pi((s_n, \delta(q_0, L(\rho))))$.
  If $\pi$ is an optimal policy for $\calM$, then $\pi'$ is an optimal policy that minimizes the expected dissatisfaction score, \ie, the solution to \PPwPPoTLO.
\end{theorem}

\begin{proof}
We  establish a connection between the expected dissatisfication score of $\pi'$ and the value of $\pi$.
First, we use (\ref{eq:policy_sat_degree}) to expand the expected dissatisfication score of $\pi'$:
  \begin{equation}
  \label{eq:pi_prim_value}
     \dos(\pi',\varphi)
     = \sum_{\rho \in S^\ast} \Pr(\rho;M^{\pi'})\cdot \dos(L(\rho),\varphi)  
  \end{equation}
  Next, we expand the value of $\pi$.
  \begin{align}
   & J_{\pi}(v_0)
    = \sum_{h \in V^\ast} \Pr(h; \calM^\pi)\cdot R(h) \nonumber  \\
    & = \sum_{h \in V^\ast: R(h) = 0} \Pr(h;\calM^\pi)\cdot 0 \nonumber \\ 
    &    + \sum_{h \in V^\ast: R(h) \neq 0} \Pr(h; \calM^\pi) \cdot R(h) 
  \end{align}
  Using the result of Lemma~\ref{lem:d_and_R}, we write this summation as:
  \begin{align}
    J_{\pi}(v_0) = &   \sum\limits_{\substack{h \in V^\ast:  R(h) \neq 0}} \Pr(h; \calM^\pi) \cdot (1 -\dos(L(\mathsf{Proj}_S(h)))) \nonumber \\
   &\cdot (opt(\varphi)+1)) \nonumber \\
   & \text{[ \textsl{replace $opt(\varphi)+1)$ by $K$, $\Pr(h; \calM^\pi)$ by $\Pr^\pi(h)$}.]}\nonumber\\
  = &   (K \cdot \sum\limits_{\substack{h \in V^\ast:  R(h) \neq 0}} \Pr^\pi(h) \nonumber \\
   & - K \cdot \sum\limits_{\substack{h \in V^\ast:  R(h) \neq 0}} \Pr^\pi(h)\dos(L(\mathsf{Proj}_S(h))) \nonumber \\
  = &   K\cdot \sum\limits_{\substack{h \in V^\ast}} \Pr^\pi(h) - K\cdot \sum\limits_{\substack{h \in V^\ast:  R(h) = 0}} \Pr^\pi(h) \nonumber \\
   & -  K\cdot \sum\limits_{\substack{h \in V^\ast:  R(h) \neq 0}} \Pr^\pi(h)  \dos(L(\mathsf{Proj}_S(h))   \label{eq:1} \\
= &   (K  -   K\cdot \sum\limits_{\substack{h \in V^\ast:  R(h) = 0}} \Pr^\pi(h) \cdot 1 \nonumber \\
& -   K\cdot \sum\limits_{\substack{h \in V^\ast:  R(h) \neq 0}} \Pr^\pi(h)  \dos(L(\mathsf{Proj}_S(h)) \label{eq:2} \\
= &   K  -  
 K \cdot \sum\limits_{\substack{h \in V^\ast} }\Pr^\pi(h) \cdot  \dos(L(\mathsf{Proj}_S(h)) ) \label{eq:3}
  \end{align}

  From \eqref{eq:1} to \eqref{eq:2}, we use the probability axiom that $\sum_{h\in V^\ast }\Pr(h; \calM^\pi)=1$. From \eqref{eq:2} to \eqref{eq:3}, we use Lemma~\ref{lem:d_and_R} that if $R(h)=0$ then  $\dos(L(\mathsf{Proj}_S(h)))=1$. 
  Thus, relating \eqref{eq:3} and \eqref{eq:pi_prim_value}, we have
  \begin{align}
   & J_{\pi}(v_0) = K- K\cdot \dos(\pi',\varphi),  
  \end{align}
  and therefore $\argmax_\pi J_\pi(v_0)  = \argmin_{\pi'} \dos(\pi',\varphi) $, that is, a policy $\pi$ that maximizes $J$ yields a policy $\pi'$ that minimizes the dissatisfaction score $\dos$.
\end{proof}

\section{Complexity Analysis}
The first step of the algorithm constructs a \ac{wdfa} that encodes $\varphi$. The constructed \ac{dfa} from \ac{ltlf} formulas is double-exponential in the size of the formulas in the worst case \cite{wolperConstructingAutomataTemporal2001a,de2021compositional}.  However, in practice this translation is tractable for commonly seen \ac{ltlf} formulas in robotic planning. 
%
The construction of automata for ordered disjunction $\varphi_1\orderedOR \varphi_2$ and prioritized conjunction $\varphi_1\&\varphi_2$ using Def.~\ref{def:wdfa_orderOR} and Def.~\ref{def:wdfa_prioritizedAnd}, respectively, takes a polynomial time to the sizes of the \ac{wdfa}'s for sub-formulas $\varphi_1$ and $\varphi_2$.
Constructing the product MDP $\calM$ takes a polynomial time to the size of the \ac{wdfa} and the original \ac{mdp}.
And computing an optimal policy for $\calM$ takes a  time polynomial in the size of the product \ac{mdp} $\calM$, using standard techniques (value/policy iteration or linear programming).

\section{Experiment}

\label{sec:experiment}

We show the efficacy  of the proposed algorithm using several examples of probabilistic robotic motion planning. 
%
\footnote{All experiments are executed on an Ubuntu 20.04 machine with AMD Ryzen 9 5900X CPU and 32 GB RAM. We use the Gurobi solver for planning in \ac{mdp}. The computational times of solving the optimal planning problem for any $8 \times 8$ gridworlds with different formulas are no more than $0.1$ seconds. The code can be found in the supplementary file.}

\begin{figure}[ht]
    \centering    \includegraphics[width=0.75\linewidth]{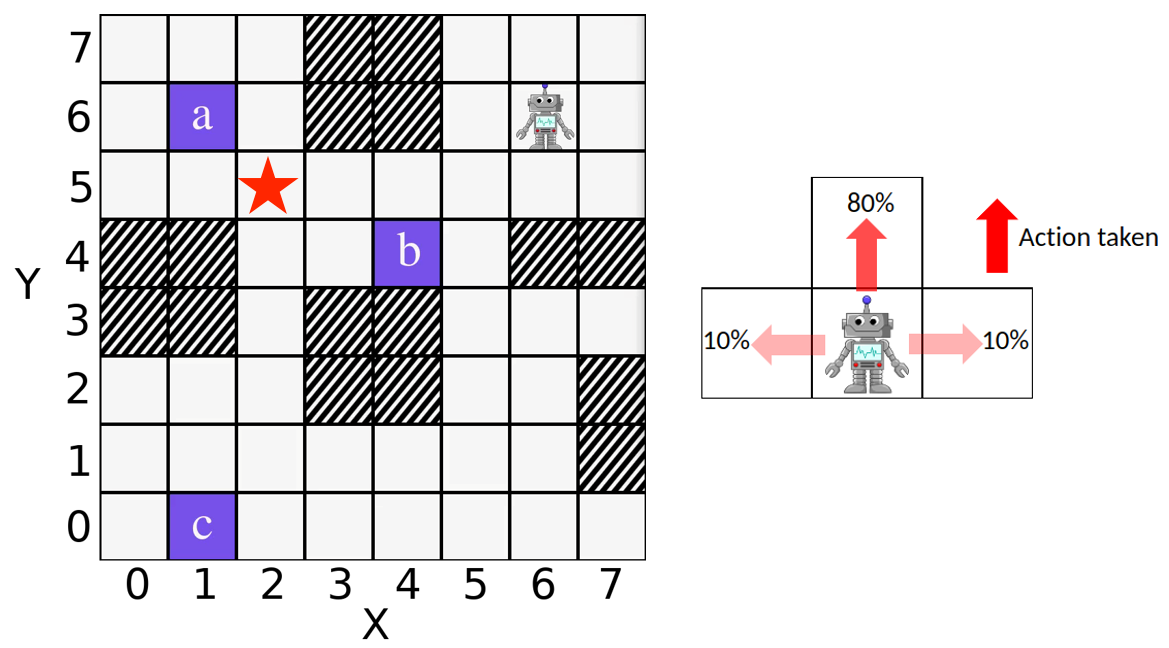}
    \caption{A $8 \times 8$ stochastic gridworld $g_1$ (without the red star) and the transition probabilities when an action ``N'' is taken. The red star is an additional hole introduced in gridworld $g_2$. }
    \label{fig:grid_world_1} 
    \vspace{-2ex}
\end{figure}
Consider a small stochastic gridworld $g_1$ shown in Fig.~\ref{fig:grid_world_1}. For each state $s \in S$, the robot has four actions: ``N'', ``W'', ``S'', ``E''. After taking an action from a state, the robot transits to the \emph{intended} cell with probability $0.8$ and slips to \emph{unintended} cells with probability $0.1$. If the robot takes an action and reaches the boundary wall, then it stays in the original cell.  
The initial state of the robot is 
$(6, 6)$. The shaded areas denote holes. Once the robot enters holes, it gets stuck. 
%
Regions of interest are labeled $a$, $b$, and $c$. Accordingly, $\calAP = \{a, b, c \}$. Each of these atomic propositions holds at a time instant when the robot is in the region labeled by the corresponding atomic proposition.

 
Given the set $\calAP$ of atomic propositions, we consider the following  preference formula   (see Appendix~\ref{app:example} for the \ac{wdfa}.)
\[
    \Eventually b \orderedOR (\Eventually a \lor \Eventually c).
\]

We computed the optimal policy $\pi_{\ast}$ that minimizes the expected dissatisfaction score.
To see the difference of ordered disjunction and regular disjunction, we also compute a optimal policy that maximizes the probability of satisfying formula $ \Eventually b \lor (\Eventually a \lor \Eventually c)$. 
We denote this policy as $\pi_{\lor}$. We plot the optimal values for different initial states in Fig~\ref{fig:case_1_normal}.
%
Then we perform policy evaluation of $\pi_{\ast}$ and $\pi_{\lor}$ against $\Eventually b$ and $\Eventually a \lor \Eventually c$, separately. The probabilites of satisfying formula $\phi$ for $\phi \in \{ \Eventually b,\Eventually a \lor \Eventually c \}$  for different initial states are shown in Fig.~\ref{fig:case_1_dis_b}, ~\ref{fig:case_1_dis_a_c}, ~\ref{fig:case_1_normal_b} and ~\ref{fig:case_1_normal_a_c}.
 Comparing Fig.~\ref{fig:case_1_dis_b} against Fig.~\ref{fig:case_1_normal_b} for the formula $\Eventually b$, $\pi_{\ast}$ achieves higher values in the most areas of the gridworld, especially at the top left corner. 
On the other side, comparing Fig.~\ref{fig:case_1_dis_a_c} against Fig.~\ref{fig:case_1_normal_a_c}, we can see that $\pi_{\lor}$ achieves higher probability of satisfying $\Eventually a \lor \Eventually c$ than that of policy $\pi_\ast$ in \deleted{the} most areas of the gridworld, especially at the top right corner. This comparison indicates that  when $\Eventually b$ is preferred to $\Eventually a \lor \Eventually c$, the preference-based policy   gravitates towards satisfying $\Eventually b$. 

\begin{figure}[ht]
    \centering    \vspace{-2ex}
    \begin{subfigure}[b]{0.48\linewidth}
    \includegraphics[trim={40 20 60 40},clip,width=\linewidth]{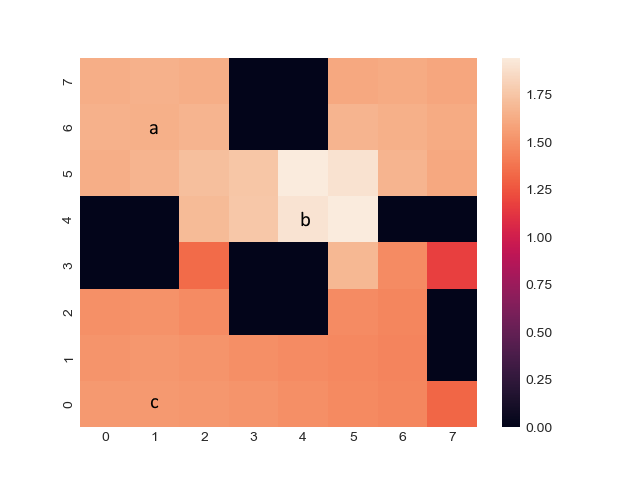}
    \subcaption{$\langle \pi_{\ast}, \Eventually b \orderedOR (\Eventually a \lor \Eventually c) \rangle$}
    \label{fig:case_1_dis}
    \end{subfigure}
    \hfill
    \begin{subfigure}[b]{0.48\linewidth}
    \includegraphics[trim={40 20 60 40},clip,width=\linewidth]{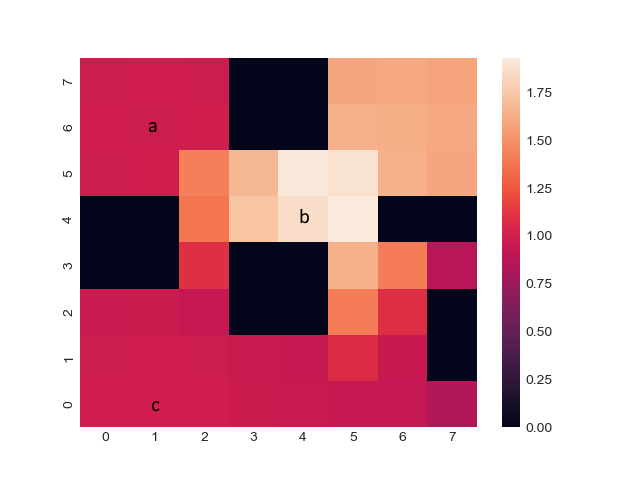}
    \subcaption{$\langle \pi_{\lor}, \Eventually b \lor (\Eventually a \lor \Eventually c) \rangle$}
    \label{fig:case_1_normal}
    \end{subfigure}
    \begin{subfigure}[b]{0.48\linewidth}
        \includegraphics[trim={40 20 60 40},clip,width=\linewidth]{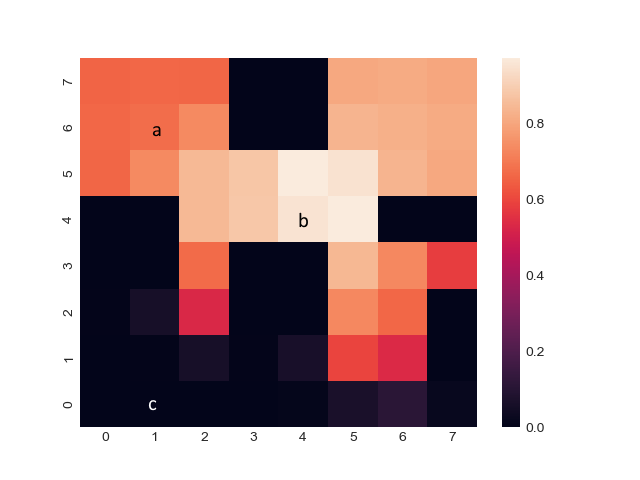}
     \subcaption{$\langle\pi_{\ast}, \Eventually b\rangle$}
    \label{fig:case_1_dis_b}
    \end{subfigure}
    \hfill
    \begin{subfigure}[b]{0.48\linewidth}
        \includegraphics[trim={40 20 60 40},clip,width=\linewidth]{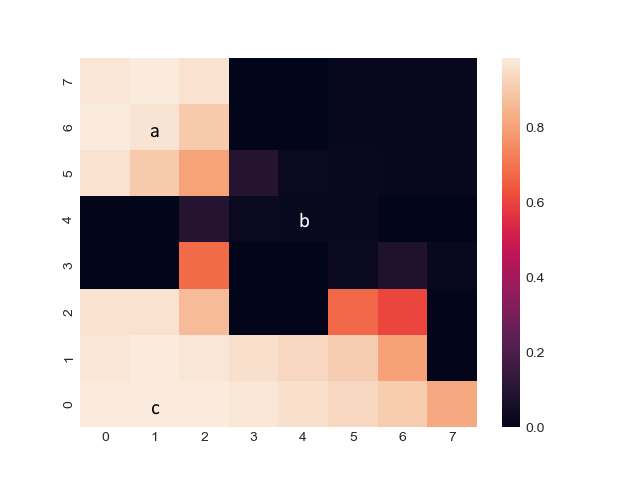}
    \subcaption{$\langle\pi_{\ast}, \Eventually a \lor \Eventually c\rangle$}
    \label{fig:case_1_dis_a_c}
    \end{subfigure}
    \begin{subfigure}[b]{0.48\linewidth}
        \includegraphics[trim={40 20 60 40},clip,width=\linewidth]{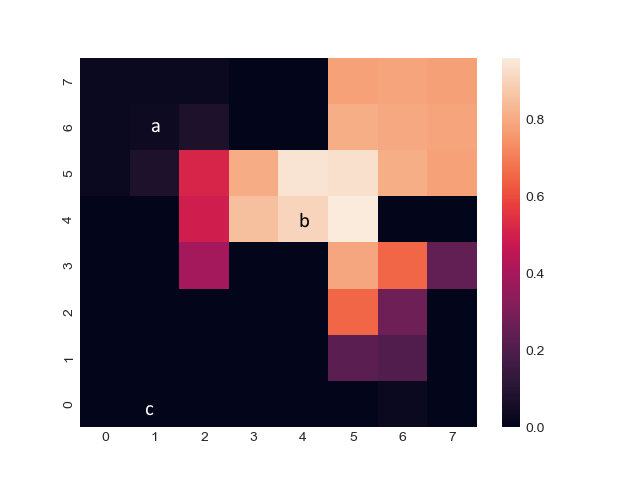}
    \subcaption{$\langle\pi_{\lor}, \Eventually b\rangle$}
    \label{fig:case_1_normal_b}
    \end{subfigure}
    \hfill
    \begin{subfigure}[b]{0.48\linewidth}
        \includegraphics[trim={40 20 60 40},clip,width=\linewidth]{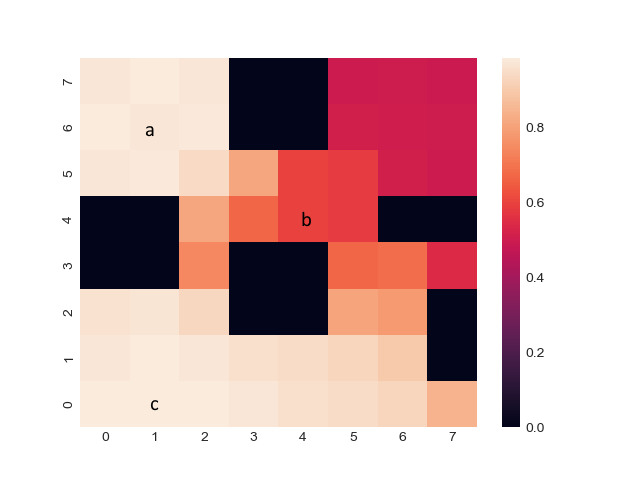}
    \subcaption{$\langle\pi_{\lor}, \Eventually a \lor \Eventually c\rangle$}
    \label{fig:case_1_normal_a_c}
    \end{subfigure}
    \vspace{-2ex}
    \caption{Each subfigure with the subcaption $\langle\pi, \phi\rangle$ is the value $J((\cdot, \vec{q}_0); \langle\pi, \phi\rangle)$ of policy evaluation of policy $\pi$ given the formula $\phi$  starting from different initial states in gridworld $g_1$. } 
    \label{fig:case_1_heatmaps}
\end{figure}

Next, we consider the following formula that has prioritized conjunction and nested ordered disjunctions:
$
    \varphi_3 = \varphi_1 \prioritizedAnd \varphi_2,$
where $\varphi_1 = \Eventually b \orderedOR (\Eventually a \lor \Eventually c)$ and $\varphi_2 = \Eventually(a \land \Eventually(b \land \Eventually c)) \orderedOR \Eventually (a \land \Eventually c) \lor \Eventually (b \land \Eventually c)$. This task formula describes that the system needs to satisfy $\varphi_1$ and $\varphi_2$ both, with $\varphi_1$ having a higher priority than $\varphi_2$. 

For this case, we consider an additional gridworld $g_2$ which includes an additional hole at the position $(2,5)$, which blocks the access to $a$. 
Given the formula $\varphi_3$, we compute the optimal policies $\pi_{\ddagger}$ when \deleted{the} region $a$ is accessible and $\pi_{\ddagger}^{-}$ when \deleted{the} region $a$ is inaccessible. 
%

\begin{figure}[H]
    \centering
    \begin{subfigure}[c]{0.23\textwidth}
        \includegraphics[trim={40 20 60 40},clip,width=\linewidth]{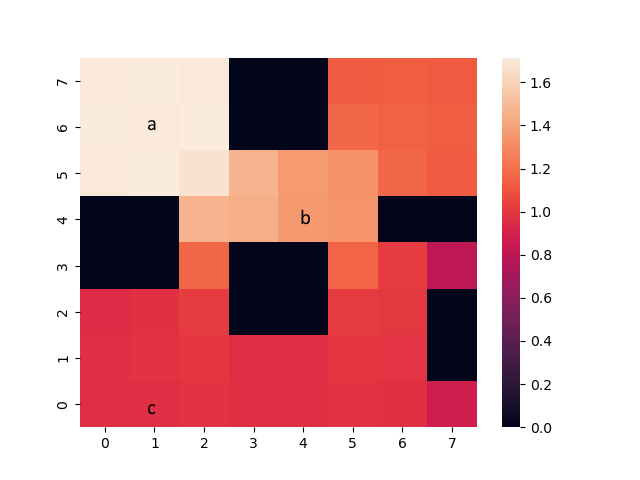}
        \subcaption{}
        \label{fig:varphi_3_value}
    \end{subfigure}
    \hfill
    \begin{subfigure}[c]{0.23\textwidth}
        \includegraphics[trim={40 20 60 40},clip,width=\linewidth]{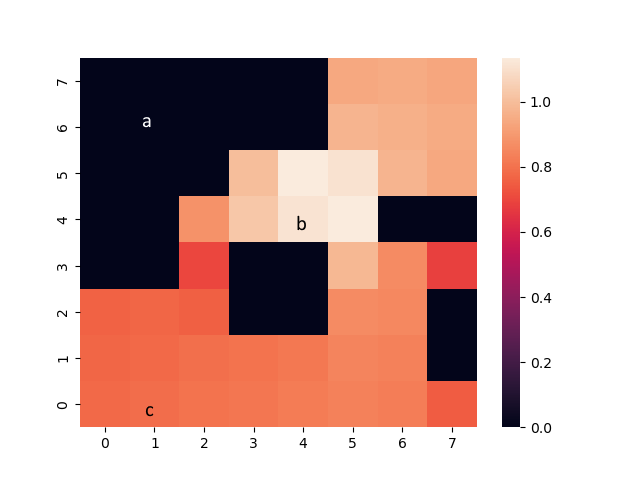}
        \subcaption{}
        \label{fig:varphi_4_value}
    \end{subfigure}    \vspace{-2ex}

    \caption{(a) The value $J((\cdot, \Vec{q}_0); \langle\pi_{\ddagger}, \varphi_3\rangle)$ in gridworld $g_1$. (b) The value $J((\cdot, \Vec{q}_0); \langle\pi_{\ddagger}^{-}, \varphi_3 \rangle)$ in gridworld $g_2$.}
\end{figure} 
We plot the heatmaps of $\varphi_3$ for these two gridworlds in Fig.~\ref{fig:varphi_3_value} and~\ref{fig:varphi_4_value}. The following observation is made: When $a$ is accessible, starting from the  upper left corner, the agent receives higher values with the optimal policy. But if $a$ is not accessible, the upper left corner states have values zero. This is because the formulas $\varphi_2$ cannot be satisfied   as the agent cannot reach region $c$ when starting from the upper left corner. Therefore, $\varphi_3$ is not satisfiable. The state values under the optimal policy given $a$ accessible are higher than the state values when $a$ is not accessible, indicating the agent can achieve a more preferred outcome in the gridworld $g_1$.

\section{Conclusion}
In this paper, we introduced a new specification language, termed prioritized qualitative choice linear temporal logic on finite traces (\ac{pqcltlf}), for compactly specifying a temporal goal along with the user's preferences on sub-goals. We presented an automatic translation from this language to weighted deterministic finite automaton. We used this translation in solving the problem of computing a policy that minimizes the expected dissatisfaction score of a given \ac{pqcltl} formula in a stochastic environment modeled by an MDP. By bridging the gap between preferences in AI and temporal logic planning, this work enables future study that incorporates preference elicitation and learning from positive/negative data and adaptive planning in sequential decision-making problems.
%
%

\bibliographystyle{named}
\bibliography{ijcai23}
\clearpage
\thispagestyle{empty}
\newpage
\appendix

  \section{Proofs}
\label{app:proofs}

Proof of Lemma 2

\begin{proof} The run $\rho$ has a weight 1 if and only if  $q_{n-1}\in F$, $\sigma_{n-1}=\ltimes $, and  $q_n=\sink$.
  Due to the acceptance condition for \ac{dfa}s, the word $\word(\rho)=\sigma_0 \sigma_1\ldots \sigma_{n-1}$ is accepted and thus satisfies the \ac{ltlf} formula $\varphi.$
\end{proof}

Proof of Lemma 3

\begin{proof}Given the penultimate state $\vec{q}_{n-1}= (q_1,q_2)$, if $\weight_1(q_1,\ltimes,\sink)=\weight_2(q_2,\ltimes, \sink)=0$, then $\word(\rho) \not \models \varphi$, meaning it does not satisfy $\varphi$. 
If $\weight_1(q_1,\ltimes,\sink)=0$, but $\weight_2(q_2,\ltimes, \sink)>0$, then $\word(\rho)$ satisfies $\varphi_2$ to a positive degree but does not satisfy $\varphi_1$. 
The satisfaction degree w.r.t. $\varphi_1\orderedOR \varphi_2$ is the sum of the satisfaction degree w.r.t. $\varphi_2$ and the optionality of $\varphi_1$. 
Else, if $\weight_1(q_1,\ltimes,\sink)>0$, then the satisfaction degree w.r.t. $\varphi_1 \orderedOR \varphi_2$ is the satisfaction degree w.r.t. $\varphi_1$.
\end{proof}

Proof  of Lemma 4

\begin{proof}
    The proof is by construction and similar to the proof of Lemma~\ref{lma:orderedDisjunct}. Thus, it is omitted.    
\end{proof}

Proof  of Lemma 5

\begin{proof}
  The property can be shown based on the recursive definition. First, it is clear that if the \ac{pqcltlf} formula is an \ac{ltlf} formula, then the optionality is one and the maximal weight of all defined transitions is one. Consider two \ac{pqcltlf} formulas $\varphi_1,\varphi_2$, and their corresponding \ac{wdfa}s $\calA =\langle Q_i \cup \{\sink\},\Sigma \cup \{
  \ltimes\}, \delta_i, q_{0i}, \weight_i \rangle $ that satisfies $opt(\varphi_i) = \max\{\weight_i(q,a,q')\mid \delta_i(q,a,q') \text{ is defined.}\}$.
  
  In the \ac{wdfa} of  the ordered disjunction $\varphi_1\orderedOR \varphi_2$, the maximal weight by construction is $\max\{\weight_2(q_2,\ltimes, \sink) + opt(\varphi_1)\} = \max\{\weight_2(q_2,\ltimes, \sink)\} + opt(\varphi_1) =  opt(\varphi_2) + opt(\varphi_1)$, which is consistent with Def.~\ref{def:optionality}.
  
  In the \ac{wdfa} of  the prioritized conjunction $\varphi_1\& \varphi_2$, the maximal weight by construction is $\max_{i,j}\{opt(\varphi_2)\times (i-1)+j\}$ where $0<i \le \max\{\weight_1(q_1,\ltimes, \sink)\} = opt(\varphi_1)$ and $0<j \le\max\{\weight_2(q_2,\ltimes, \sink)\} = opt(\varphi_2)  $. Therefore,   $\max_{i,j}\{opt(\varphi_2)\times (i-1)+j\} = opt(\varphi_2)\times (opt(\varphi_1)-1) + opt(\varphi_1) = opt(\varphi_2)\times opt(\varphi_1)$. This is again consistent with Def.~\ref{def:optionality}.
\end{proof}

Proof of Lemma 6.
\begin{proof}
  A finite run $\rho = v_0a_0 v_1a_1\ldots v_n$ receives a nonzero reward only if there exists $0\le k\le n$, $v_k  = (s_{\bot}, \sink)$, and for all $j \le k$, $v_j \neq (s_{\bot}, \sink)$. The total reward of $\rho$ is upper bounded by $opt(\varphi)$. Therefore, for any policy $\pi$ and any state $v$, the limit of $J_{\pi, N}(v)$ as $N\rightarrow \infty$ exists and is upper bounded by $opt(\varphi)$.
\end{proof}

Proof of Lemma 7.

\begin{proof}
  We show that for every improper, optimal Markovian policy, there is a proper, Markovian policy that obtains the same value.
  Consider an improper, optimal Markovian policy $\pi^\dagger$ under which there is an infinite run. Since the reward is only obtained by reaching state $(s_{\bot},\sink)$, the  infinite run $h$ will have a reward of zero.
  Thus, a proper policy $\pi^\ast$ that has the same value $J_{\pi^\dagger}(v)= J_{\pi^\ast}(v)$ can be constructed by copying $\pi^\dagger$ for all finite runs. For all infinite runs, $\pi^\ast$ is obtained from $\pi^\dagger$ by terminating at any state with a zero reward.
\end{proof}

Proof of Lemma 8.

 \begin{proof} Prove by construction. For the first case, let us recall $R(h) = \sum_{i}^{n-1}R(v_i, a_i)$. If $L(\mathsf{Proj}_S (h)) \models_k \varphi$ for some $k>0$, then $R(h) = opt(\varphi) -  \weight(q_n, \ltimes, \sink)+1$.
Plug in $R(h)$, and we have $\dos(L(\mathsf{Proj}_S (h)), \varphi) =1-\frac{R(h)}{opt(\varphi)+1} = \frac{\weight(q_n, \ltimes, \sink)}{opt(\varphi)+1}$,  complying with Lemma~\ref{lma:ltlf}, ~\ref{lma:orderedDisjunct}, and~\ref{lma:prioritizedConj}. For the second case, if  $L(\mathsf{Proj}_S (h)) \not \models \varphi$, then $R(h=0)$, then $\dos(L(\mathsf{Proj}_S (h)) = 1$, complying with Def.~\ref{def:dissatisfaction_score}. 
\end{proof}
 
\section{Example of weighted automata construction}
\label{app:example}

We illustrate the construction of \ac{wdfa} using an example.
\begin{example}
  \label{example:orderedOR}
Given two \ac{ltlf} formulas $\Eventually b$ and $ \Eventually a \lor \Eventually c$ and a \ac{pqcltlf} formula $ \Eventually b \orderedOR ( \Eventually a \lor \Eventually c)$, reading ``if possible, eventually satisfy $b$, and if not possible, eventually satisfy $a$ or $c$.'' The \ac{wdfa}s for the \ac{ltlf} formulas are shown in Fig.~\ref{fig:Fb} and~\ref{fig:FaORc}, and the \ac{wdfa} is shown in Fig.~\ref{fig:varphi_1}. For clarity, we use propositional logic formulas instead of $2^\calAP$ as the symbols for the transitions.  
For example, $b \wedge (a\lor c): 0$ stands for $\{b, a\}:0, \{b, c\}:0, $ and $\{b, a, c\}:0$. 
From Fig.~\ref{fig:varphi_1}, we see that the weight transits from $(0, 1)$ to $\sink$ is $2$, that is because by triggering that transition the satisfied formula $\Eventually a \lor \Eventually c$ is less preferred.
  \begin{figure}[ht]
    \centering
    \begin{subfigure}[b]{0.23\textwidth}
    \begin{tikzpicture}[->,>=stealth',shorten >=1pt,auto,node distance=2.5cm,
        scale = 0.7,transform shape]

  \node[state,initial] (0) {$0$};
  \node[state] (1) [above right of=0] {$1$};
  \node[state,accepting] (sink) [below right of=1] {$\sink$};

  \path (0) edge              node {$b:0$} (1)
        (0) edge              node {$\ltimes:0$} (sink)
        (1) edge              node {$\ltimes:1$} (sink)
        (0) edge   [loop below]           node {$\neg b:0$} (0)
        (1) edge   [loop above]           node {$\top:0 $} (1);

\end{tikzpicture}
    \vspace{-2em}
    \subcaption{}
    \label{fig:Fb}
    \end{subfigure}
    \hfill
    \begin{subfigure}[b]{0.23\textwidth}
        \begin{tikzpicture}[->,>=stealth',shorten >=1pt,auto,node distance=2.5cm,
        scale = 0.7,transform shape]

  \node[state,initial] (0) {$0$};
  \node[state] (1) [above right of=0] {$1$};
  \node[state,accepting] (sink) [below right of=1] {$\sink$};

  \path (0) edge              node {$a \lor c:0$} (1)
        (0) edge              node {$\ltimes:0$} (sink)
        (1) edge              node {$\ltimes:1$} (sink)
        (0) edge   [loop below]           node {$\neg a \land \neg c:0$} (0)
        (1) edge   [loop above]           node {$\top:0$} (1);

\end{tikzpicture}
        \vspace{-2em}
        \subcaption{}
        \label{fig:FaORc}
    \end{subfigure}
    \begin{subfigure}[b]{0.3\textwidth}
        \begin{tikzpicture}[->,>=stealth',shorten >=1pt,auto,node distance=2.5cm, scale = 0.7,transform shape]

  \node[state,initial] (0) {$0, 0$};
  \node[state] (1) [above right of=0] {$1, 0$};
  \node[state] (2) [below right of=0] {$0, 1$};
  \node[state] (3) [below right of=1] {$1, 1$};
  \node[state,accepting] (sink) [right of=3] {$\sink$};

  \path 
    (0) edge              node {$b\land (a\lor c):0$} (3)
  (0) edge              node {$b \wedge \neg a \wedge \neg c: 0$} (1)
  (0) edge              node {$\neg b \wedge (a \lor c):0$} (2)
  (1) edge              node {$a \lor c:0$} (3)
  (2) edge              node[right] {$b:0$} (3)
  (2) edge[loop below]              node[right] {$\neg b:0$} (2)
  (1) edge[loop above]              node {$\neg a \wedge \neg c:0$} (1)
  (3) edge[loop     right ]        node {$\top:0$} (3)
  
  (1) edge[bend left]              node {$\ltimes:1$} (sink)
  (2) edge[bend right]              node {$\ltimes:2$} (sink)
  (0) edge[bend right=100,looseness=1.6,out=270]              node[below] {$\ltimes:0$} (sink)
  (3) edge   [bend left]           node[left] {$\ltimes:1$} (sink);
\end{tikzpicture}
        \vspace{-3em}
    \subcaption{}
    \label{fig:varphi_1}
    \end{subfigure}
    \caption{(a) The \ac{wdfa} accepting the formula $\Eventually b$. (b) The \ac{wdfa} accepting the formula $\Eventually a \lor \Eventually c$. (c) The \ac{wdfa} accepting the formula $\Eventually b \orderedOR (\Eventually a \lor \Eventually c)$.}
    \label{fig:wdfa}
\end{figure}

In Table~\ref{tab:degrees_of_satisfaction} we list the satisfaction degrees given different words.
%
%
%
From the dissatisfaction scores, we have  $w_1 \succeq^\phi w_2 \succeq^\phi w_3$, where $\phi=\Eventually b \orderedOR (\Eventually a \lor \Eventually c)$.
\begin{table*}[!htb]
    \centering
    \caption{Dissatisfaction Scores for Words \emph{w.r.t.} $\Eventually b \orderedOR (\Eventually a \lor \Eventually c)$}
    \resizebox{0.85\linewidth}{!}{
    \begin{tabular}{c|l|cccc}
        \hline
        Words & Word descriptions       & $w \models_i \Eventually b$ & $ w \models_j \Eventually a \lor \Eventually c$ & $ w \models_k \Eventually b \orderedOR (\Eventually a \lor \Eventually c) $ & $\dos(w, \Eventually b \orderedOR (\Eventually a \lor \Eventually c))$ \\ \hline
        $w_1$ &  $\{b\}\{a\}\ltimes$       & $1$                       & $1$                        & $1$ & $1/3$ \\
        $w_2$ &  $\emptyset \emptyset \{a\} \ltimes $ & $\not \models $                       & $1$                        & $2$ &  $2/3$                                               \\
        $w_3$ & $\emptyset \emptyset \ltimes $         & $\not \models $                       & $\not \models $                        & $\not \models $ & $1$                                                 \\
        \hline
    \end{tabular}
    }
    \label{tab:degrees_of_satisfaction}
\end{table*}
\end{example}

\end{document}